\documentclass[final,3p,times]{elsarticle}

\usepackage{amsmath,amssymb,amsfonts,amscd,amsthm}
\usepackage{graphicx}

\usepackage{color,soul}

\def\C{\mathbb C}

\renewcommand{\i}{{\mathrm i}}

\def\S{\mathcal S}
\def\rank{\mathrm{rank}}

\newtheorem{lem}{Lemma}[section]

\newtheorem{thm}[lem]{Theorem}

\newtheorem*{thm0}{Theorem}

\theoremstyle{definition}

\newtheorem{rem}[lem]{Remark}

\begin{document}

\begin{frontmatter}

\title{Tripartite connection condition for quantum graph vertex}
%
\author[label1]
{Taksu Cheon\corref{cor1}}
\ead{taksu.cheon@kochi-tech.ac.jp}
\author[label2,label3]
{Pavel Exner} 
\ead{exner@ujf.cas.cz}
\author[label1]
{Ond\v rej Turek}
\ead{ondrej.turek@kochi-tech.ac.jp}
\address[label1]
{Laboratory of Physics, Kochi University of Technology,
Tosa Yamada, Kochi 782-8502, Japan}
\address[label2]{Doppler Institute for Mathematical Physics and Applied
Mathematics, Czech Technical University,
B\v rehov{\'a} 7, 11519 Prague, Czechia}
\address[label3]{Department of Theoretical Physics, Nuclear Physics
Institute, Czech Academy of Sciences,
25068 \v{R}e\v{z} near Prague, Czechia}
\cortext[cor1]{corresponding author}

\date{\today}
\begin{abstract}
We discuss formulations of boundary conditions in a quantum graph vertex and demonstrate that the so-called $ST$-form can be further reduced up to a form more effective in certain applications:  In particular, in identifying the number of independent parameters for given ranks of two connection matrices, or in calculating the scattering matrix when both matrices are singular.  The new form of boundary conditions, called the $PQRS$-form, also gives a natural scheme to design generalized low and high pass quantum filters.
\end{abstract}
\begin{keyword}
Schrodinger operator \sep singular vertex
\sep boundary conditions
\PACS 03.65.-w \sep 03.65.Db \sep 73.21.Hb 
%
%
\end{keyword}

\end{frontmatter}


\section{Introduction}

Quantum graphs are becoming increasingly relevant as mathematical models of quantum wire based single electron devices. At the heart of quantum
graph is the behavior of quantum particle at a graph vertex, in general connecting $n$ graph edges, which can be regarded as a natural generalization
of singular point interaction in one dimension \cite{AG05}. At a glance it
is simple, but in reality a highly nontrivial object.

General mathematical characterizations of vertex couplings have been there for more than two decades. While at first general theory of self-adjoint extensions was used and the corresponding boundary conditions were worked
out for particular cases \cite{ES89}, in 1999 general conditions were written
in the form $A\Psi+B\Psi'=0$ by Kostrykin and Schrader \cite{KS99} with 
elegantly formulated requirements on $A$ and $B$. It
includes situations when one or both the matrices $A,B$ are singular; for
those cases alternative descriptions were developed \cite{Ku04, FKW07} which
employ projections to complements of the rank of these matrices.

One may wonder why the physical contents of the vertex couplings, including the singular cases, is of interest -- recall that most existing
models employ the most simple free coupling, often called Kirchhoff. The
main reason is that it can give us alternative means to control transport
through such graph structures which is the ultimate practical goal of these
investigations. No less important is that it gives theoretical tools to
analyze various classes of graphs -- recall, e.g., the use of
scale-invariant boundary conditions in investigation of radial tree
graphs, see \cite{NS00, SS02} and subsequent work of other authors.

Attempts to understand physical meaning of vertex coupling take different
routes. Some are ``constructive'', trying to approximate vertex with a
prescribed coupling by a family of graphs \cite{CS98, ENZ01, CE04, ET07, CET10} 
or various ``fat graphs'', see \cite{KZ01, EP09} and references therein. An alternative is
to look into scattering properties associated with a particular coupling
an to try to classify their type. Such a study was undertaken, in particular, in the article \cite{CT10}.

A drawback of the conditions $A\Psi+B\Psi'=0$ is that the matrix pair
$(A,B)$ determining the vertex coupling
is not unique. 
Our starting point, in this paper, is a particular unique
version of them called the $ST$-form, which was developed in \cite{CET10}, in which the matrices $A,B$ exhibit a specific rank-based reduction and contain two parametric submatrices, $S$ and $T$. Its
properties were further investigated in \cite{CT10}: the key point in this paper
is the observation that, at $k\to\infty$, scattering matrix is reduced to
the one obtained from scale-invariant coupling which generalizes the one
studied for a particle on a line by F\"ul\"op and Tsutsui \cite{FT00}, and also by
Solomyak and coauthors \cite{NS00, SS02}.
Namely, at this limit, the interaction is specified
only by $T$ and with the influence of $S$ vanishing asymptotically. At the
same time the opposite asymptotics, $k\to 0$, yields scattering matrix
reduced to the one obtained from what is in \cite{CT10} called reverse
F\"ul\"op-Tsutsui condition, being specified only by $\tilde{T}$ 
with the coupling matrix $\tilde{S}$ 
influencing only the second term of the asymptotics.

In this work, we show that there is another useful and unique form 
for $A$ and $B$, which we call $PQRS$ form, that lays a bridge between $ST$ and 
reverse $ST$-forms, and demonstrate its connections to the other unique
ways to write the coupling. 
In particular, we determine the number of
independent parameters which characterize classes of singular couplings
with ranks of $A$ and $B$ fixed. This new form turns out to be very 
useful in specifying 
F\"ul\"op-Tsutsui and reverse 
F\"ul\"op-Tsutsui forms at small and large $k$ limits. Also it is shown that 
certain ``zero-limits'' of $PQRS$-form lead to formulae that amount to the 
generalization of the classification of $n=3$ singular vertex \cite{CET09},
for which ``$Y$-junction'' can function as a spectral branching filter.

\section{Motivation}

\subsection{$ST$-form and its relation to the scattering matrix}

Generally, for the boundary conditions
\begin{equation}
\label{AB}
A\Psi+B\Psi'=0\,,
\end{equation}
the scattering matrix is given by the formula
\begin{eqnarray}
\label{S_AB}
\S(k)=-(A+\i kB)^{-1}(A-\i kB)
\end{eqnarray}
and thus its computation needs to invert a matrix $A+\i kB$ which is of the size $n\times n$. However, if one of the matrices $A,B$ has not full rank $n$, the size of the matrix to be inverted can be reduced. Let us demonstrate it below.

For any value of $r_B=\rank(B)$, any admissible boundary condition for a singular vertex 
in quantum graph~\eqref{AB}
can be equivalently expressed in the $ST$-form
\begin{eqnarray}
\label{ST}
\begin{pmatrix}
I^{(r_B)} & T \\
0 & 0
\end{pmatrix}
\Psi'=
\begin{pmatrix}
S & 0 \\
-T^* & I^{(n-r_B)}
\end{pmatrix}
\Psi
\end{eqnarray}
for certain $S$ and $T$, 
where the symbol $I^{(j)}$ denotes the identity matrix of size $j\times j$. 
In this formalism, the scattering matrix $\S(k)$ acquires the form
\begin{eqnarray}
\label{S(S,T,k)}
\S(k)=
&&\!\!\!\!\!\!\!\!\!
-I^{(n)}+
2
\begin{pmatrix}
\left(I^{(r_B)}+TT^*-\frac{1}{\i k}S\right)^{-1} & 
\left(I^{(r_B)}+TT^*-\frac{1}{\i k}S\right)^{-1}T \\
T^*\left(I^{(r_B)}+TT^*-\frac{1}{\i k}S\right)^{-1} & 
T^*\left(I^{(r_B)}+TT^*-\frac{1}{\i k}S\right)^{-1}T
\end{pmatrix}
\nonumber \\
=
&&\!\!\!\!\!\!\!\!\!
-I^{(n)}+
2\begin{pmatrix} I^{(r_B)} \\ T^* \end{pmatrix}
\left(I^{(r_B)}+TT^*-\frac{1}{\i k}S\right)^{-1}
\begin{pmatrix} I^{(r_B)} & T\end{pmatrix}
\end{eqnarray}
which is easier to be calculated than~\eqref{S_AB}, since one has to perform an inversion for a matrix $r_B\times r_B$.

Moreover, $\S(k)$ given by~\eqref{S(S,T,k)} can be expanded for high energies $k\gg1$:
Since the matrix $\left(I^{(m)}+TT^*-\frac{1}{\i k}S\right)^{-1}$ satisfies
\begin{eqnarray}
\left(I^{(r_B)}+TT^*-\frac{1}{\i k}S\right)^{-1}=
&&\!\!\!\!\!\!\!\!\!
\left[\left(I^{(r_B)}+TT^*\right)\left(I^{(r_B)}-\left(I^{(r_B)}+TT^*\right)^{-1}\frac{1}{\i k}S\right)\right]^{-1}
\nonumber \\
=
&&\!\!\!\!\!\!\!\!\!
\left(I^{(r_B)}+TT^*\right)^{-1}+\sum_{j=1}^{\infty}\left(\frac{1}{\i k}\right)^j\left[\left(I^{(r_B)}+TT^*\right)^{-1}S\right]^j
\left(I^{(r_B)}+TT^*\right)^{-1} ,
\end{eqnarray}
we have
\begin{eqnarray}
\S(k)=
&&\!\!\!\!\!\!\!\!\!
-I^{(n)}+
2\begin{pmatrix}
I^{(r_B)} \\ T^*
\end{pmatrix}
\left(I^{(r_B)}+TT^*\right)^{-1}
\begin{pmatrix}
I^{(r_B)} & T
\end{pmatrix}
\nonumber \\
&&\!\!\!\!\!\!\!\!
+ 
2\begin{pmatrix}
I^{(r_B)} \\ T^*
\end{pmatrix}
\sum_{j=1}^{\infty}\left(\frac{1}{\i k}\right)^j\left[\left(I^{(r_B)}+TT^*\right)^{-1}S\right]^j
\left(I^{(r_B)}+TT^*\right)^{-1}
\begin{pmatrix}
I^{(r_B)} & T
\end{pmatrix}\,,
\end{eqnarray}
and in particular we see that
\begin{eqnarray}
\lim_{k\to\infty}\S(k)=
-I^{(n)}+ 2
\begin{pmatrix}
\left(I^{(r_B)}+TT^*\right)^{-1} & \left(I^{(r_B)}+TT^*\right)^{-1}T \\
T^*\left(I^{(r_B)}+TT^*\right)^{-1} & T^*\left(I^{(r_B)}+TT^*\right)^{-1}T
\end{pmatrix} \, ,
\end{eqnarray}
i.e. the scattering matrix corresponding to the vertex coupling expressed in the form~\eqref{ST} tends to the scattering matrix of the \emph{scale invariant} vertex coupling expressed by~\eqref{ST} with $S=0$. The effect of matrix $S$ in~\eqref{ST} thus fades away for $k\to\infty$.

The $ST$-form itself does not allow us to expand $\S(k)$ at the same time at $k=0$ except for special cases when the submatrix $S$ is regular. If this expansion is required, we need to transform the $ST$-form into its reverse form
\begin{eqnarray}
\label{reverseST}
\begin{pmatrix}
I^{(r_A)} & \tilde{T} \\
0 & 0
\end{pmatrix}
\Psi=
\begin{pmatrix}
\tilde{S} & 0 \\
-\tilde{T}^* & I^{(n-r_A)}
\end{pmatrix}
\Psi'
\end{eqnarray}
where $r_A=\rank(A)$ and $\tilde{S},\tilde{T}$ are properly chosen matrices.
In a similar manner to the case of $ST$-form, we can find that
\begin{eqnarray}
\label{RS(S,T,k)}
\S(k)=
I^{(n)}-
2\begin{pmatrix} I^{(r_A)} \\ \tilde{T^*} \end{pmatrix}
\left(I^{(r_A)}+\tilde{T}\tilde{T^*}-\i k \tilde{S} \right)^{-1}
\begin{pmatrix} I^{(r_A)} & \tilde{T} \end{pmatrix}\,.
\end{eqnarray}
It is easy to see that the reverse $ST$-form allows one to expand $\S(k)$ at $k=0$, but generally not at $k\to\infty$, and also enables to find the zero-momentum limit which is given by
\begin{eqnarray}
\lim_{k\to0}\S(k)=
I^{(n)}- 2
\begin{pmatrix}
\left(I^{(r_A)}+\tilde{T}\tilde{T^*}\right)^{-1} & 
\left(I^{(r_A)}+\tilde{T}\tilde{T^*}\right)^{-1}\tilde{T} \\
\tilde{T^*} \left(I^{(r_A)}+\tilde{T}\tilde{T^*}\right)^{-1} & 
\tilde{T^*} \left(I^{(r_A)}+\tilde{T}\tilde{T^*}\right)^{-1}\tilde{T} 
\end{pmatrix}\,.
\end{eqnarray}

We conclude that both the $ST$-form and its reversed version generally simplify the matrix inversion needed for computation of $\S(k)$, but neither of them makes it possible to expand $\S(k)$ for $k\gg1$ and at the same time around $k=0$, except for special cases when $S$, $\tilde{S}$ are regular.

\subsection{$ST$-form and number of parameters of vertex couplings}

It follows from the $ST$-form of boundary conditions that
if $r_B<n$, then the number of real numbers parametrizing
the family of vertex couplings in a vertex of degree $n$ is reduced from $n^2$ to at most
$n^2-(n-r_B)^2$, cf.~\cite{CET10}.

At the same time, if the boundary conditions are transformed into the reverse $ST$-form~\eqref{reverseST}, one can notice that the number of parameters is bounded above by the value $n^2-(n-r_A)^2$, since this is the total number of free real parameters involved in $\tilde{S}$ and $\tilde{T}$.

There is a natural question on the actual number of free parameters if both $r_B,r_A$ are less than $n$. This question cannot be anwered just with the help of the $ST$-form or its reverse, for that purpose we need to develop another form of boundary conditions, 
which we shall consider in the next section.


\section{$PQRS$-form}\label{PQRSsec}

In the previous section we have come across two problems to them the $ST$-form gives only a partial answer. The reason why the $ST$-form does not lead to the full solutions lies in the fact that it is asymetric with respect to $\rank(A),\rank(B)$: whereas $\rank(B)$ substantially determines its structure, cf.~\eqref{ST}, the value of $\rank(A)$ plays no significant role. In this section we introduce a symmetrized version of the $ST$-form in which both ranks are essentially equally important. The new form of boundary conditions, we will call it \emph{$PQRS$-form}, will then help us to solve the two foregoing problems, namely
\begin{itemize}
\item to find the exact number of free parameters if both $\rank(B)$, $\rank(A)$ are fixed,
\item to expand $\S(k)$ at both $k\to\infty$ and $k=0$ at the same time.
\end{itemize}
The formulation of the $PQRS$-form of boundary conditions follows.

\begin{thm}\label{PQRS-dec}
Let us consider a quantum graph vertex of a degree $n$.
\begin{itemize}
\item[(i)] If $0\leq r_A\leq n$, $0\leq r_B\leq n$, $S\in\C^{r_A+r_B-n,r_A+r_B-n}$ is a self-adjoint matrix
and $P\in\C^{r_A+r_B-n,n-r_B}$, $Q\in\C^{n-r_A,n-r_B}$, $R\in\C^{n-r_A,r_A+r_B-n}$, then the equation
\begin{equation}
\label{PQRS}
\begin{pmatrix}
I^{(r_A+r_B-n)} & 0 & P \\
R & I^{(n-r_A)} & Q \\
0 & 0 & 0
\end{pmatrix}
\Psi'=
\begin{pmatrix}
S & -SR^* & 0 \\
0 & 0 & 0 \\
-P^* & (RP-Q)^* & I^{(n-r_B)}
\end{pmatrix}
\Psi \,.
\end{equation}
expresses admissible boundary conditions. 
\item[(ii)] For any vertex coupling there exist numbers $0\leq r_A\leq n$, $0\leq r_B\leq n$
and a numbering of edges such that the coupling is described by
the boundary conditions \eqref{PQRS} with the uniquely given
matrices $P\in\C^{r_A+r_B-n,n-r_B}$, $Q\in\C^{n-r_A,n-r_B}$, $R\in\C^{n-r_A,r_A+r_B-n}$ and a \emph{regular} self-adjoint matrix $S\in\C^{r_A+r_B-n,r_A+r_B-n}$.
\item[(iii)] Consider a quantum graph vertex of degree $n$ with
the numbering of the edges explicitly given; then there is a
permutation $\Pi\in S_n$ such that the boundary conditions may be
written in the modified form
\begin{equation}
\label{CouplingPQRS}
\begin{pmatrix}
I^{(r_A+r_B-n)} & 0 & P \\
R & I^{(n-r_A)} & Q \\
0 & 0 & 0
\end{pmatrix}
\tilde{\Psi}'=
\begin{pmatrix}
S & -SR^* & 0 \\
0 & 0 & 0 \\
-P^* & (RP-Q)^* & I^{(n-r_B)}
\end{pmatrix}
\tilde{\Psi}
\end{equation}
for
\begin{eqnarray}
\tilde{\Psi}=
\begin{pmatrix}
\psi_{\Pi(1)}(0)\\
\vdots\\
\psi_{\Pi(n)}(0)
\end{pmatrix}
\qquad \tilde{\Psi}'=
\begin{pmatrix}
\psi_{\Pi(1)}'(0)\\
\vdots\\
\psi_{\Pi(n)}'(0)
\end{pmatrix},
\end{eqnarray}
where the regular self-adjoint matrix $S\in\C^{m,m}$ and the matrices 
$P\in\C^{r_A+r_B-n,n-r_B}$, $Q\in\C^{n-r_A,n-r_B}$, $R\in\C^{n-r_A,r_A+r_B-n}$ depend unambiguously on $\Pi$. This formulation
of boundary conditions is in general not unique, since there may
be different admissible permutations $\Pi$, but one can make it
unique by choosing the lexicographically smallest possible permutation
$\Pi$.
\end{itemize}
\end{thm}

\begin{proof}
We start with the claim (ii).
Consider boundary conditions given in the $ST$-form
\begin{eqnarray}
\begin{pmatrix}
I^{(r_B)} & T_{ST} \\
0 & 0
\end{pmatrix}
\Psi'=
\begin{pmatrix}
S_{ST} & 0 \\
-T_{ST}^* & I^{(n-r_B)}
\end{pmatrix}
\Psi
\end{eqnarray}
where $r_B=\rank(B)\leq n$, $S_{ST}\in\C^{m,m}$ is a self-adjoint matrix
and $T_{ST}\in\C^{m,n-m}$ is a general matrix.

If we denote $r_A=\rank(A)$, we see that $r_A=\rank(S_{ST})+n-r_B$, hence
\begin{eqnarray}
\rank(S_{ST})=r_A+r_B-n\,.
\end{eqnarray}
We may suppose without loss of generality that the first $r_A+r_B-n$ ($=\rank(S_{ST})$) rows of $S_{ST}$ are linearly independent and the remaining $n-r_A$ rows are their linear combinations. If it is not the case, it obviously suffices to apply a simultaneous permutation on first $r_B$ rows and columns of both matrices $A$ and $B$ and renumber the components of $\Psi$, $\Psi'$ in the same manner.
Now we decompose both matrices $A$, $B$ in the following way:
\begin{eqnarray}
\label{ST3}
\begin{pmatrix}
I^{(r_A+r_B-n)} & 0 & T_1 \\
0 & I^{(n-r_A)} & T_2 \\
0 & 0 & 0
\end{pmatrix}
\Psi'=
\begin{pmatrix}
S_{11} & S_{21}^* & 0 \\
S_{21} & S_{22} & 0 \\
-T_1^* & -T_2^* & I^{(n-r_B)}
\end{pmatrix}
\Psi
\end{eqnarray}
where
\begin{eqnarray}
\begin{pmatrix}
T_1 \\
T_2
\end{pmatrix}
=T_{ST}\,,\qquad
\begin{pmatrix}
S_{11} & S_{21}^* \\
S_{21} & S_{22}
\end{pmatrix}
=S_{ST}
\end{eqnarray}
and the sizes of all submatrices are determined by the blocks $I^{(r_A+r_B-n)}$, $I^{(n-r_A)}$ and $I^{(n-r_B)}$.
Since the rows of $(S_{21}\; S_{22})$ are linear combinations of those of $(S_{11}\; S_{21}^*)$ (which are linearly independent), there is a unique matrix $-R\in\C^{n-r_A,r_A+r_B-n}$ such that
\begin{eqnarray}
\label{defR}
\left(S_{21} \; S_{22}\right)=-R\left(S_{11} \; S_{21}^*\right)\,.
\end{eqnarray}
In the next step we multiply the system \eqref{ST3} from the left by the matrix
\begin{eqnarray}
\begin{pmatrix}
I^{(r_A+r_B-n)} & 0 & 0 \\
R & I^{(n-r_A)} & 0 \\
0 & 0 & 0
\end{pmatrix}
\end{eqnarray}
to obtain
\begin{eqnarray}
\label{prePQRS}
\begin{pmatrix}
I^{(r_A+r_B-n)} & 0 & T_1 \\
R & I^{(n-r_A)} & T_2+RT_1 \\
0 & 0 & 0
\end{pmatrix}
\Psi'=
\begin{pmatrix}
S_{11} & S_{21}^* & 0 \\
0 & 0 & 0 \\
-T_1^* & -T_2^* & I^{(n-r_B)}
\end{pmatrix}
\Psi\,.
\end{eqnarray}
We notice that \eqref{defR} gives an explicit relation between $S_{21}$ and $S_{11}$ via the matrix $R$, namely
\begin{eqnarray}
S_{21}=-RS_{11}\,.
\end{eqnarray}
We employ this fact to eliminate $S_{21}^*$ from \eqref{prePQRS}, then we set $T_2+RT_1=Q$ and rename $T_1$ as $P$ and $S_{11}$ as $S$. Herewith we arrive at the sought final form of boundary conditions~\eqref{PQRS}.

It follows from the construction that the matrix $S\in\C^{r_A+r_B-n,r_A+r_B-n}$ is self-adjoint and regular, and $P\in\C^{r_A+r_B-n,n-r_B}$, $Q\in\C^{n-r_A,n-r_B}$, $R\in\C^{n-r_A,r_A+r_B-n}$ are general matrices of given sizes.

Thereby (ii) is proved. Since the claim (iii) can be obtained immediately from (ii) using a
simultaneous permutation of elements in the vectors $\Psi$ and
$\Psi'$, it remains to prove (i). We have
to show that the matrices
\begin{eqnarray}
A=-
\begin{pmatrix}
S & -SR^* & 0 \\
0 & 0 & 0 \\
-P^* & (RP-Q)^* & I^{(n-r_B)}
\end{pmatrix}
\quad\text{and}\quad
B=
\begin{pmatrix}
I^{(r_A+r_B-n)} & 0 & P \\
R & I^{(n-r_A)} & Q \\
0 & 0 & 0
\end{pmatrix}
\end{eqnarray}
satisfy the condition $\rank(A|B)=n$ and that $AB^*$ is self adjoint. Both can be verified in a straightforward way.
\end{proof}

\begin{rem}\label{regS}
If the block with the matrix $S$ is present in the $PQRS$-form (i.e. if $r_A+r_B-n>0$), then it is supposed to be regular. This assumption could be in fact dropped, but we would lose the uniqueness of $R$ then, cf. \eqref{defR}.
\end{rem}
%
%
In the following sections,
we shall demonstrate several applications of the $PQRS$-form.
%

\section{Number of parameters of vertex couplings}

The whole family of vertex couplings in a vertex of degree $n$ may be decomposed into disjoint subfamilies according to the pair $(\rank(A),\rank(B))$; the number of the subfamilies equals $\frac{(n+1)(n+2)}{2}$ by virtue of the condition $\rank(A|B)=n$. Such a decomposition is useful for a study of physical properties of quantum graph vertices: In~\cite{CET09}, a classification of vertex couplings based on the values $\rank(A)$, $\rank(B)$ has been provided for $n=3$, and in Section~\ref{Y} of this paper we extend the ideas to a general $n$.

Each subfamily given by the pair $(\rank(A),\rank(B))$ has certain number of real parameters that is easily determined with the help of the $PQRS$-form: If we just sum up the number of real parameters of the matrices $P$, $Q$, $R$, $S$ involved in~\eqref{PQRS}, we arrive after a simple manipulation at
\begin{eqnarray}
\label{paramnmb}
n^2-(n-r_A)^2-(n-r_B)^2\,.
\end{eqnarray}
This formula shows in a very clear way how the number of parameters of the vertex coupling decreases with decreasing ranks of $A$ and $B$.

\section{Relations to other parametrizations}

The fact that Kostrykin-Schrader conditions $A\Psi+B\Psi'=0$ 
are non-unique inspired various other ways how to write the 
coupling. A commonly used one employs matrices which are functions of
a given unitary $n\times n$ matrix $U$, namely
 \begin{eqnarray} 
 \label{Harmer}
(U-I)\Psi+i(U+I)\Psi'=0\,.
 \end{eqnarray}
In the quantum graph context it was proposed in \cite{H00, KS00},
however, it was known much earlier in the general theory of
boundary value problems \cite{GG91}.

As mentioned in the introduction, alternate conditions using projections
were developed for situations when the matrices in (\ref{AB}) can be singular.
The paper \cite{Ku04} dealt with the case when one matrix is singular, 
general conditions of this type allowing for singularity in both 
matrices were formulated in \cite{FKW07}. Let us recall this result:

\begin{thm0}\nonumber (\cite{FKW07}) For any vertex coupling in a vertex of degree $n$ there are two orthogonal and mutually orthogonal projectors $\mathcal{P}, \mathcal{Q}$ operating in $\C^n$ and an invertible self-adjoint operator $\Lambda$ acting on the subspace $\mathcal{C}\C^n$, where $\mathcal{C}=1-\mathcal{P}-\mathcal{Q}$, such that the boundary conditions can be expressed by the system of equations
\end{thm0}
\begin{subequations}
\begin{eqnarray}
\mathcal{P}\Psi=0\,, \label{projP} \\
\mathcal{Q}\Psi'=0\,, \label{projQ} \\
\mathcal{C}\Psi'=\Lambda \mathcal{C}\Psi\,. \label{projC}
\end{eqnarray}
\end{subequations}

Naturally, different unique descriptions of the coupling are 
mutually related. For instance, it is obvious that the projections
$\mathcal{P}, \mathcal{Q}$ correspond to eigenspaces of $U$ with
the eigenvalues $\mp 1$, respectively, and $\Lambda$ is the part 
of $U$ in the orthogonal complement to them. What is more relevant
here is that Theorem ({\it \cite{FKW07}}) is tightly connected to the $PQRS$ form. Indeed, it apparently holds
\begin{itemize}
\item Eq.~\eqref{projP} corresponds to the $n-r_B$ lines $\left(-P^* \quad (RP-Q)^* \quad I^{(n-r_B)}\right)\Psi=0$ of~\eqref{PQRS}\,,
\item Eq.~\eqref{projQ} corresponds to the $n-r_A$ lines $\left(R \quad I^{(n-r_A)} \quad 0\right)\Psi'=0$ of~\eqref{PQRS}\,.
\end{itemize}
In other words, $\mathcal{P}$ and $\mathcal{Q}$ are projectors on the subspaces generated by the columns of
\begin{eqnarray}
\begin{pmatrix}
-P \\
RP\!-\!Q \\
I^{(n-r_B)}
\end{pmatrix}
\qquad\text{and}\qquad
\begin{pmatrix}
R^* \\
I^{(n-r_A)} \\
Q^*
\end{pmatrix}\,.
\end{eqnarray}

It may not be completely obvious that the different coupling 
classes are characterized by the same number of parameters; note 
that the difference between (\ref{paramnmb}) and the number of parameters 
of the matrix $\Lambda$ equal to $(r_A+r_B-n)^2$ is given by 
\begin{eqnarray} 
\label{extraparam}
\Delta_{A,B} = 2\left[ r_A r_B - (r_A+r_B-n)^2 \right] \,.
\end{eqnarray}
The fact that the difference is positive unless $r_A=r_B=n$ is
due the different setting of the boundary conditions. The PQRS form
works with a partly fixed basis while (\ref{Harmer}) does not, hence
we need extra parameters to fix the ranges of $\mathcal{P}$ and 
$\mathcal{Q}$. We will employ the following elementary result.

\begin{lem}
The number of real parameters required to fix an $M$-dimensional subspace 
of $\mathbb{C}^N$ equals $2M(N-M)$.
\end{lem}
\begin{proof}
Note that the expression must be symmetric w.r.t. the interchange 
$M \leftrightarrow N-M$. To determine such a subspace we have to fix 
$N-M$ complex components of the $M$ vectors spanning it, 
which gives the result.
\end{proof}

To get the desired conclusion we apply the lemma twice: first to $n-r_A$ 
vectors spanning the complement to the range of $A$ in $\mathbb{C}^n$,
and then to $n-r_B$ vectors spanning $(\mathrm{Ran}\,B)^\perp =
\mathrm{Ker}\,B^*$ in the remaining $r_A$-dimensional space; it yields
\begin{eqnarray}
2r_A(n-r_A) + 2(n-r_B)(r_A+r_B-n) = \Delta_{A,B}\,,
\end{eqnarray}
hence the numbers of parameters are indeed the same.

\section{Scattering matrix and its expansions for $k\to0$ and $k\to\infty$}

Let us proceed to the scattering matrix expressed in terms of the submatrices $P,Q,R,S$ appearing in the $PQRS$-form.
%
%
In order to have the formula for $\S(k)$ in more compact form, we introduce the following auxiliary matrix $n\times r_A+r_B-n$:
\begin{eqnarray}
X=
\begin{pmatrix}
I^{(r_A+r_B-n)} \\
0 \\
P^*
\end{pmatrix}
-
\begin{pmatrix}
R^* \\
I^{(n-r_A)} \\
Q^*
\end{pmatrix}
\left(I^{(n-r_A)}+RR^*+QQ^*\right)^{-1}(R+QP^*)\,.
\end{eqnarray}
Then a straighforward calculation leads to the following expression for $\S(k)$:
\begin{eqnarray}
\label{S(P,Q,R,S,k)}
\S(k)=
-I^{(n)}+
2
\begin{pmatrix}
R^* \\
I^{(n-r_A)} \\
Q^*
\end{pmatrix}
\left(I^{(n-r_A)}+RR^*+QQ^*\right)^{-1}
\begin{pmatrix}
R & I^{(n-r_A)} & Q
\end{pmatrix}
+
2X\left(X^*X-\frac{1}{\i k}S\right)^{-1}X^*\,.
\end{eqnarray}

\begin{rem}
Formula~\eqref{S(P,Q,R,S,k)} can be in some sense regarded as an explicit version of the ``projector'' formula $\S(k)=-\mathcal{P}+\mathcal{Q}-(\Lambda-\i k)^{-1}(\Lambda+\i k)\mathcal{C}$ found in~\cite{FKW07}. The first two projectors $\mathcal{P},\mathcal{Q}$ have been discussed in the last section, and the third one, $\mathcal{C}$, can be shown to be the orthogonal projector on the subspace generated by the columns of $X$, and thus the term $2X\left(X^*X-\frac{1}{\i k}S\right)^{-1}X^*$ from~\eqref{S(P,Q,R,S,k)} is equal to $-2\i k(\Lambda-\i k)^{-1}\mathcal{C}$.
\end{rem}

In the rest of the section we will calculate the expansions of $\S(k)$ for high and low energies. Let us consider boundary conditions expressed in the $PQRS$-form~\eqref{PQRS}. We suppose that the block $S$ is present (i.e. $\rank(A)+\rank(B)-n>0$); if it is to the contrary, the vertex coupling is scale-invariant and thus independest of $k$.

Similarly as the $ST$-form, the $PQRS$-form allows us to expand $\S(k)$ for high energies,
\begin{eqnarray}
\S(k) = -I^{(n)} + 2
\begin{pmatrix}
R^* \\
I^{(n-r_A)} \\
Q^*
\end{pmatrix}
\left(I^{(n-r_A)}+RR^*\!+\!QQ^*\right)^{-1}
\begin{pmatrix}
R & I^{(n-r_A)} & Q
\end{pmatrix}
+2X\left(X^*X\right)^{-1}X^*
\nonumber \\
+2X\sum_{j=1}^{\infty}\left(\frac{1} {\i k}\right)^j\left[\left(X^*X\right)^{-1}S\right]^j\cdot\left(X^*X\right)^{-1}X^*\,,
\end{eqnarray}
and hence to find the limit of $\S(k)$ for $k\to\infty$,
\begin{eqnarray}
\label{highkS}
\lim_{k\to\infty}\S(k)
= I^{(n)}-
2
\begin{pmatrix}
-P \\
RP-Q \\
I^{(n-r_B)}
\end{pmatrix}
\left[I^{(n-r_B)}+P^*P+(RP\!-\!Q)^*(RP\!-\!Q)\right]^{-1}
\begin{pmatrix}
-P^* & P^*R^*\!\!-\!Q^* & I^{(n-r_B)}
\end{pmatrix}
\end{eqnarray}
(here we have used the identity $\mathcal{P}+\mathcal{Q}+\mathcal{C}=I$ from~\cite{FKW07}).

The advantage of the $PQRS$-form is that one can at the same time obtain the expansion of $\S(k)$ around $k=0$. It suffices to realize that
(note that the matrix $S$ is supposed to be regular, cf. Remark~\ref{regS})
\begin{eqnarray}
\left(X^*X-\frac{1}{\i k}S\right)^{-1}=
\left[\frac{\i}{k}S\left(I^{(r_A+r_B-n)}-\i kS^{-1}X^*X\right)\right]^{-1}
=
-\i k \sum_{j=0}^{\infty}(\i k)^j\left(S^{-1}X^*X\right)^j\cdot S^{-1}
\,;
\end{eqnarray}
then the sought expansion of $\S(k)$ at $k=0$ equals
\begin{eqnarray}
\S(k) = -I^{(n)} + 2
\begin{pmatrix}
R^* \\ I^{(n-r_A)} \\ Q^*
\end{pmatrix}
\left(I^{(n-r_A)}+RR^*+QQ^*\right)^{-1}
\begin{pmatrix}
R & I^{(n-r_A)} & Q
\end{pmatrix}
-2\i kX\left[\sum_{j=0}^{\infty}(\i k)^{j}\left(S^{-1}X^*X\right)^j\right]S^{-1}X^*\,.
\end{eqnarray}
In particular we have
\begin{eqnarray}
\label{lowkS}
\lim_{k\to0}\S(k)=
-I^{(n)}+
2
\begin{pmatrix}
R^* \\
I^{(n-r_A)} \\
Q^*
\end{pmatrix}
\left(I^{(n-r_A)}+RR^*+QQ^*\right)^{-1}
\begin{pmatrix}
R & I^{(n-r_A)} & Q
\end{pmatrix} \, .
\end{eqnarray}
%

\section{Generalized spectral branching filter}\label{Y}
We want to show that the $PQRS$-parametrization is suited to classify singular vertex
in terms of $\delta$ and $\delta^\prime$ connections, since it gives a convenient expression
for the scattering matrix at both $k \to 0$ and $k \to \infty$ limits.
%
%
Let us assume that all elements of $P$ are given by $p$, $Q$, by $q$, and $R$, by $r$,
respectively, where $p$, $q$ and $r$ are taken to be real numbers.
Namely, we set
\begin{eqnarray}
\label{stfrm}
P = p F^{(r_A+r_B-n, n-r_B)}, \quad
Q = q F^{(n-r_A, n-r_B)}, \quad
R = r F^{(n-r_A, r_A+r_B-n)} ,
\end{eqnarray}
where $F^{(m,l)}$ is the matrix of $m$ rows and $l$ columns, that is, of size $l \times m$, 
all of whose elements are equal to $1$.
The F\"ul\"op-Tsutsui limit~\eqref{highkS} together with the identity
\begin{eqnarray}
\left(I^{(m)}+\alpha F^{(m,m)}\right)^{-1} = I^{(m)} - \frac{\alpha}{1+\alpha m} F^{(m,m)}
\end{eqnarray}
allows us to express the scattering amplitudes between any elements of the block $\mu$ and $\nu$ ($\mu,\nu\in\{1,2,3\}$), which we denote ${\cal S}_{\{\mu\}\{\nu\}}(k)$, in the form
\begin{eqnarray}
&&\!\!\!\!\!\!\!\!
\lim_{k\to \infty} |{\cal S}_{\{1\}\{2\}}(k)| =
\frac{2(n-r_B)\,|p|\,|q-(r_A+r_B-n)rp|}{1+l_p \, |p|^2+l_q \, |q-(r_A+r_B-n)rp|^2},
\nonumber \\
&&\!\!\!\!\!\!\!\!
\lim_{k\to \infty} |{\cal S}_{\{2\}\{3\}}(k)|= 
\frac{2\,|q-(r_A+r_B-n)rp|}{1+l_p \, |p|^2+l_q \, |q-(r_A+r_B-n)rp|^2},
\nonumber \\
&&\!\!\!\!\!\!\!\!
\lim_{k\to \infty} |{\cal S}_{\{3\}\{1\}}(k)| = 
\frac{2\,|p|}{1+l_p \, |p|^2+l_q \, |q-(r_A+r_B-n)rp|^2} ,
\end{eqnarray}
with $l_p = (n-r_B)(r_A+r_B-n)$ and $l_q = (n-r_B)(n-r_A)$. 
%
%
%
\begin{figure}
\center{\includegraphics[width=12.5cm]{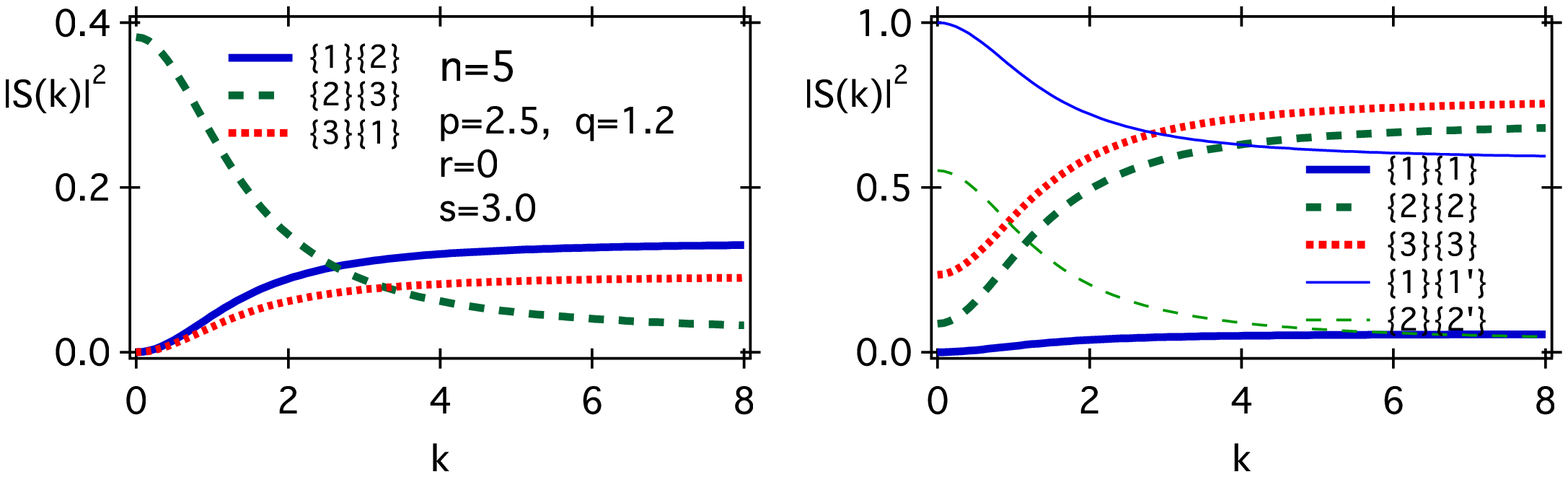}}
\label{fig1}
\caption
{
Quantum scatterings off singular vertex of degree $n=5$.
Boundary condition is given by $PQRS$-form with block devision $2-2-1$.
The block matrices $P$, $Q$, $R$, $S$ are given as constants $p$, $q$, $r$, $s$ times 1-filled matrix $F$, respectively.
Each lines represents transmission/reflection probability between/among blocks \{1\}, \{2\}, and \{3\}. 
This examples shows $\delta\delta\delta'$-type connection among blocks.
}
\end{figure}
Similarly, the scattering amplitudes between any element of the block $\mu$ and $\nu$ 
are identical at $k \to 0$, which
can be read out from the inverse F\"ul\"op-Tsutsui limit (\ref{lowkS}) as
\begin{eqnarray}
&&\!\!\!\!\!\!\!\!
\lim_{k\to 0} |{\cal S}_{\{1\}\{2\}}(k)| =  \frac{2\, |r|}{1+l_r\, |r|^2+l_q\, |q|^2},
\nonumber \\
&&\!\!\!\!\!\!\!\!
\lim_{k\to 0} |{\cal S}_{\{2\}\{3\}}(k)| = \frac{2\, |q|}{1+l_r\, |r|^2+l_q\, |q|^2},
\nonumber \\
&&\!\!\!\!\!\!\!\!
\lim_{k\to 0} |{\cal S}_{\{3\}\{1\}}(k)| =  \frac{2\, (n-r_A)\, |r| \, |q|}{1+l_r\, |r|^2+l_q\, |q|^2}  ,
\end{eqnarray}
with $l_r = (n-r_A)(r_A+r_B-n)$. 

%
%
\begin{figure}
\center{\includegraphics[width=12.5cm]{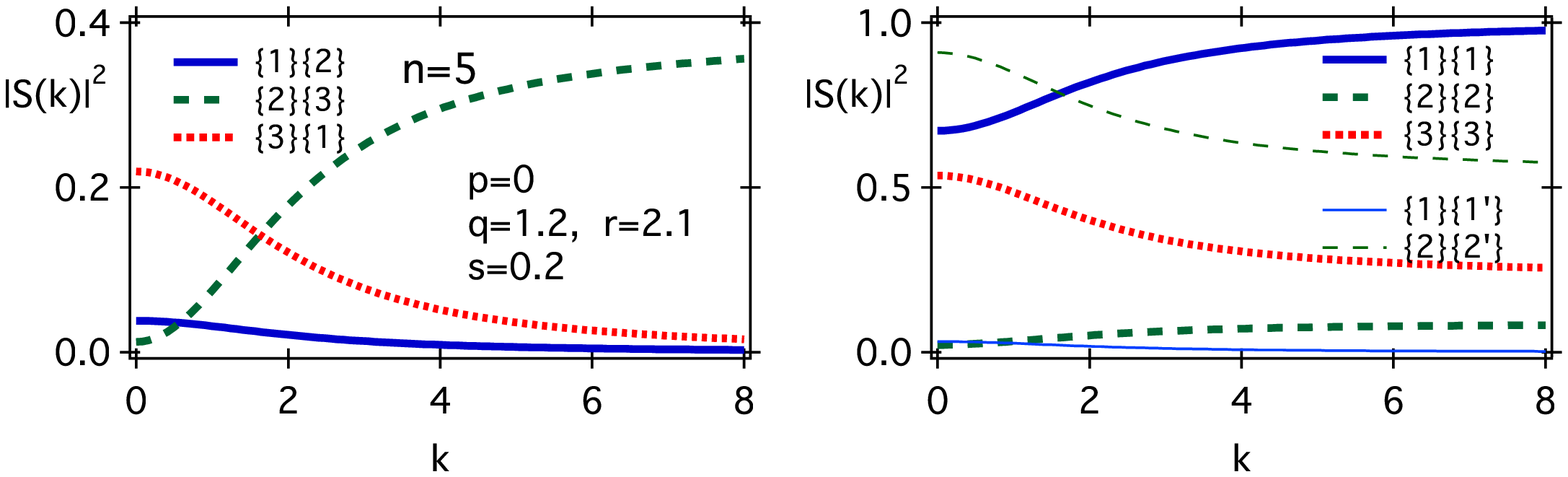}}
\label{fig2}
\caption
{
Quantum scatterings off singular vertex of degree $n=5$.
Boundary condition is given by $PQRS$-form with block devision $2-2-1$.
The block matrices $P$, $Q$, $R$, $S$ are given as constants $p$, $q$, $r$, $s$ times 1-filled matrix $F$, respectively.
Eaxh lines represents transmission/reflection probability between/among blocks \{1\}, \{2\}, and \{3\}. 
This examples shows $\delta\delta'\delta'$-type connection among blocks.
}
\end{figure}
%
These limits give us obvious ways to control the pair-wise transmission probabilities between blocks
of outgoing lines by properly tuning the absolute values of $p$, $q$ and $r$.  
Specifically, {$\delta\delta\delta^\prime$ type vertex is obtained with
\begin{eqnarray}
r=0, q\approx 1, p \gg 1
&&
\nonumber\\
\longrightarrow
&& \frac{2|p|}{1+l_q+l_p\, |p|^2}
\sim |{\cal S}_{\{3\}\{1\}}(\infty)| \sim |{\cal S}_{\{1\}\{2\}}(\infty)| \gg |{\cal S}_{\{2\}\{3\}}(\infty)| , 
\nonumber \\
&& 0= |{\cal S}_{\{3\}\{1\}}(0)| = |{\cal S}_{\{1\}\{2\}}(0)| \ll |{\cal S}_{\{2\}\{3\}}(0)| ,
\end{eqnarray}
with the moderation that $|p|$ is not too large
to keep ${\cal S}_{\{1\}\{2\}}$ $={\cal S}_{\{3\}\{1\}}$ in sizable amount.
The quantum particle entered from the lines in block $\{3\}$ is directed toward 
the lines in block $\{1\}$ when $k$ is small, and is directed toward 
the lines in block $\{2\}$ when $k$ is large, enabling the use of this connection condition as a spectral branching filter.

Similarly, {$\delta\delta^\prime\delta^\prime$ type vertex is obtained,
for example, with
\begin{eqnarray}
p=0, q\approx 1, r \gg 1
&& 
\nonumber\\
\longrightarrow
&& 0=|{\cal S}_{\{1\}\{2\}}(\infty)| = |{\cal S}_{\{3\}\{1\}}(\infty)| \ll |{\cal S}_{\{2\}\{3\}}(\infty)| , 
\nonumber \\
&& \frac{2|r|}{1+l_q+l_r\, |r|^2} 
\sim |{\cal S}_{\{1\}\{2\}}(0)| \sim |{\cal S}_{\{3\}\{1\}}(0)| \gg |{\cal S}_{\{2\}\{3\}}(0)| ,
\end{eqnarray}
with the moderation that $|r|$ is not too large
to keep ${\cal S}_{\{1\}\{2\}}$ $={\cal S}_{\{3\}\{1\}}$ in discernible size. 
In this setting, the quantum particle entered from the lines in block $\{3\}$ is directed toward the lines in block $\{2\}$ when $k$ is small, and is directed toward 
the lines in block $\{1\}$ when $k$ is large.  We can use this connection condition 
again as a spectral branching filter.
These amount to be the generalization of 
$\delta\delta\delta^\prime$ and $\delta\delta^\prime\delta^\prime$ type 
connection for $n=3$ vertex, namely the ``$Y$-junction''.

These limits are illustrated in the numerical examples with $n=5$ singular vertex
in which lines are divided into three blocks of size two, two and one with the choice of
$r_A=3$ and $r_B=4$,
which are shown in Figures 1 and 2.
In Figure 1, we display the transmission and reflection probabilities between lines in various
blocks with the choice $p=2.5$, $q=1.2$, $r=0$ and $s=3.0$, which leads to $\delta\delta\delta'$-type branching.
In Figure 2, we display the results with the choice $p=0$, $q=1.2$, $r=2.1$ and $s=0.2$, which leads to $\delta\delta'\delta'$-type branching.

%

\section{Prospects}

The results outlined in this article can serve as stepping stones for various further works and developments.
It is possible that there are choices of $P$, $Q$ and $R$ other than (\ref{stfrm}) 
that lead to simple expressions for ${\cal S}(k)$, that could help us sorting out physical contents of connection conditions further.
Up to now, the multi-vertex graphs has been considered only with ``free'' connections mostly,
or at best, with $\delta$ connections. 
Examining the system with more than two singular vertices of nontrivial characteristics should be interesting.
In this work, the physical analysis is directed to the scattering properties.
Examination of the bound state spectra, which is given as the purely imaginary poles 
of the ${\cal S}$-matrix, should be high in the list of next agenda.
Many of the singular vertex parameters are complex numbers.  The imaginary part is related to the
``magnetic'' components of the vertex coupling.
In light of the recent finding of exotic quantum holonomy in magnetic point interaction on a line \cite{TC10},
the quantum graph with magnetic vertices may be a rich play ground for phenomena related to the  quantum holonomy.

\section*{Acknowledgments}

The research was supported  by the Japanese Ministry of Education, Culture, Sports, Science and Technology under the Grant number 21540402, and also by the Czech Ministry of Education, Youth and Sports within the project LC06002.


\end{document}